 \documentclass{article}
 \usepackage{graphicx}
\usepackage{amsmath,amsthm,amsfonts,amscd,amssymb,hyperref}
\usepackage{physics}
\usepackage{authblk}

\usepackage{color}

\newtheorem{theorem}{Theorem}[section]

\newtheorem{prop}[theorem]{Proposition}

\def\phi{ \varphi }

\DeclareMathOperator{\sgn}{sgn}

\theoremstyle{definition}
\newtheorem{definition}[theorem]{Definition}
\newtheorem{example}[theorem]{Example}

\theoremstyle{remark}
\newtheorem{remark}[theorem]{Remark}

\newcommand{\ccirc}[1]{\xymatrix@1{+<1ex>[o][F-]{#1}}}

\title{$X$-States From a Finite Geometric Perspective}
\author[1]{Colm Kelleher}
\affil[1]{Institut de Math\'ematiques de Bourgogne, UMR 5584, Universit\'e de Bourgogne Franche-Comt\'e, F-21078 Dijon, France}
\author[2]{Fr\'ed\'eric Holweck}
\affil[2]{Laboratoire Interdisciplinaire Carnot de Bourgogne, ICB/UTBM, UMR 6303 CNRS, Universit\'e Bourgogne Franche-Comt\'e, F-90010 Belfort, France} 
\author[3]{P\'eter L\'evay}
\affil[3]{MTA-BME Quantum Dynamics and Correlations Research Group, Department of Theoretical Physics, Budapest University of Technology and Economics, 1521 Budapest, Hungary}
\author[4]{Metod Saniga}
\affil[4]{Astronomical Institute of the Slovak Academy of Sciences, SK-05960 Tatransk\'a Lomnica, Slovakia}
\begin{document}

\maketitle

\begin{abstract}
It is found that $15$ different types of two-qubit $X$-states split naturally into two sets (of cardinality $9$ and $6$) once their entanglement properties are taken into account. We {characterize both the validity and entangled nature of the $X$-states with  maximally-mixed subsystems in terms of certain parameters} and show that their properties are related to a special class of geometric hyperplanes of the symplectic polar space of order two and rank two. Finally, we introduce the concept of hyperplane-states and briefly address their non-local properties.
\end{abstract}

\section{Introduction}
Two-qubit $X$-states are usually introduced in the literature as two-qubit density matrices with an $X$-shape,
\begin{equation}\label{rho}
 \rho=\begin{pmatrix}
       \rho_{11} & 0 & 0 & \rho_{14}\\
       0         & \rho_{22} & \rho_{23} & 0\\
       0        & \rho_{32} & \rho_{33} & 0\\
       \rho_{41} & 0 & 0 & \rho_{44}
      \end{pmatrix}.
\end{equation}

The EPR-states and the Werner states are particular cases of $X$-states \cite{nielsen} and their properties have extensively been studied from quantum information theoretic perspective \cite{yu2005evolution, chen2011quantum}. It was A.R.P. Rau who first pointed out, in a series of papers \cite{rau2009algebraic,vinjanampathy2010generalized,rau2009mapping}, the underlying algebraic structure defining an $X$-state. Taking $X,Y,Z$ to represent the usual Pauli matrices and $I$ to be the identity matrix, Eq. (\ref{rho}) can be rewritten in the following way
\begin{equation}\label{fano}
 \rho=\dfrac{1}{4}(I\otimes I+\tau^A Z\otimes I+\tau^B I\otimes Z+\beta_{zz} Z\otimes Z+\beta_{xx} X\otimes X+\beta_{yy} Y\otimes Y+\beta_{yx} Y\otimes X+\beta_{xy} X\otimes Y,
\end{equation}
where $\tau$'s and $\beta$'s are real coefficients that can be calculated from $\rho_{ij}$. Rau also noticed that the non-trivial two-qubit operators involved in Eq. (\ref{fano}) have the algebraic structure of the projective plane of order two, the Fano plane, if one considers their products (Figure \ref{fanoplane}).
\begin{figure}
\centering
 \includegraphics[width=5cm]{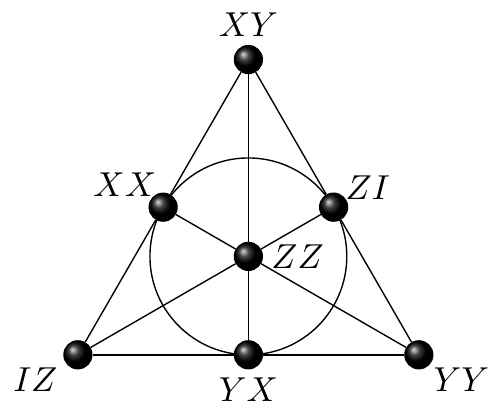}
 \caption{Up to a phase, the multiplication properties of the $7$ operators of Eq. (\ref{fano}). They form a  Fano plane. The three lines intersecting at $ZZ$ correspond to three sets of mutually commuting operators. One notices that the distinguished observable $Z\otimes Z$ commutes with all the remaining ones.
($AB$ is a short-hand for $A \otimes B$). All $15$ Fano planes of PG$(3,2)$ are explicitly listed in Appendix \ref{appA}.}\label{fanoplane}
\end{figure}

Up to a phase factor, there are $15$ non-trivial two-qubit observables and, therefore, $15$ labeled Fano planes similar to that depicted in Figure \ref{fanoplane}, each of them hosting a maximal set of observables commuting with a given one. Rau \cite{rau2009algebraic} further suggested to extend the definition of $X$-states to all two-qubit density matrices that would yield similar underlying Fano structures once decomposed on a particular two-qubit Pauli basis.

In this article we consider the $15$ different kinds of $X$-states following Rau's idea and show that they split into two groups. Group 1 consists of $X$-like states that are always separable irrespectively of the choice of parameters $\tau$'s and $\beta$'s. Group 2 entails $X$-states that can be entangled if the parameters satisfy certain conditions. The standard $X$-state given by Eq. (\ref{rho}) and Eq. (\ref{fano}) belongs to Group 2. {For Group 2 $X$-states with maximally mixed subsystems,  we also give explicit constraints on their parameters to decide on the validity and entanglement of particular states.} Then one proposes an alternative finite-geometric definition of $X$-states making use of a particular type of geometric hyperplanes -- so-called perp-sets -- of the symplectic polar space of  order two and rank two, $\mathcal{W}(3,2)$. The existence of two distinct groups of $X$-states is here embodied in the intersection properties of perp-sets with a specific hyperbolic quadric of $\mathcal{W}(3,2)$. This then naturally leads to a generalization of $X$-states for all the remaining types of hyperplanes of $\mathcal{W}(3,2)$, with subsequent analysis of their entanglement properties. Finally, we briefly address non-local properties for $X$-states of Group $2$ as well as for the other hyperplane-states that can produce entangled states.\\ \\

{\bf Symbols and Notation}. Let $\mathcal{P}_2=\{s A_1\otimes A_2, A_i\in \{I,X,Y,Z\}, s\in \{\pm 1,\pm i\}\}$ be the group of two-qubit Pauli observables where $X=\begin{pmatrix}
  0 & 1\\
                     1 & 0 
                    \end{pmatrix}$, 
$Y=\begin{pmatrix} 
0 & -i \\
i & 0    
   \end{pmatrix}$ and $Z=\begin{pmatrix} 
   1 & 0 \\ 
   0 & -1
   \end{pmatrix}$, $I$ being the identity matrix. As already mentioned, the tensor product of observables will be shorthanded
as $A_1\otimes A_2=A_1A_2$. Disregarding a phase factor, $\pm 1, \pm i$, the 15 nontrivial two-qubit Pauli operators can be identified with the 15 points of the projective space over the two element field, PG$(3,2)$, as follows. Let $A_1=Z^{\mu_1}\times X^{\nu_1}$ and $A_2=Z^{\mu_2}\times X^{\nu_2}$ be two Pauli matrices and $A_1A_2=(Z^{\mu_1}\times X^{\nu_1})(Z^{\mu_2}\times X^{\nu_2})$. Then $A_1A_2$ is a non-trivial two qubit Pauli operator iff $(\mu_1,\nu_1,\mu_2,\nu_2)$ is a nonzero vector of $V^4$, the four-dimensional vector space over the two element field GF$(2)=\{0,1\}$, and thus can be mapped to a point of PG$(3,2)$:
\begin{equation}\label{eq:pi}
 \pi:\left\{\begin{array}{lll} 
             \mathcal{P}_2\backslash I_2 & \to & \text{PG}(3,2),\\
             s(Z^{\mu_1}\times X^{\nu_1})(Z^{\mu_2}\times X^{\nu_2}) & \mapsto & [\mu_1:\nu_1:\mu_2:\nu_2].
            \end{array}\right.
\end{equation}

For instance, the observables $\{IX,-IX,iIX,-iIX\}$ are mapped to the point $[0:0:0:1]$ and $\{YZ,-YZ,iYZ,-YZ\}$ are mapped to the point $[1:1:1:0]$. A line of PG$(3,2)$ is made of triplets of points $(p,q,r)$ such that $r=p+q$. The corresponding classes of two-qubit operators $\overline{\mathcal{O}_p},\overline{\mathcal{O}}_q,
\overline{\mathcal{O}}_r$ satisfy $\overline{\mathcal{O}}_p\times \overline{\mathcal{O}}_q=\overline{\mathcal{O}}_r$. PG$(3,2)$ contains $15$ Fano planes that can neatly be parametrized by the $15$ points of PG$(3,2)$. Given a non-degenerate bilinear form \begin{equation}\label{sigma}\sigma(p,q)=p_1q_2+p_2q_1+p_3q_4+p_4q_3\end{equation} for $p=[p_1:p_2:p_3:p_4]$ and $q=[q_1:q_2:q_3:q_3] \in \text{PG}(3,2)$,  one can define the $15$ Fano planes of PG$(3,2)$ as $F_p=\{q\in \text{PG}(3,2), \sigma(p,q)=0\}$. For instance, $F_{[1:0:1:0]}=\{[1:0:1:0],[0:1:1:1],[1:1:0:1],[0:0:1:0],[1:0:0:0],[0:1:0:1],[1:1:1:1]\}$; it is easy to see that this Fano plane is the one shown in Figure \ref{fanoplane}. For the convenience of the reader, the $15$ different Fano planes of PG$(3,2)$ are given in terms of operators in Appendix \ref{appA}.

\section{Entangled $X$-states split into two groups}\label{entanglement}
Given a two-qubit system with density matrix $\rho$, its `entanglement status'  can be discerned from the fact whether the partial transpose $\rho^{\Gamma}$ is positive-semidefinite, or not \cite{horodecki2001separability}, where the partial transpose can be taken over any subsystem (PPT criterion). Namely, if $\rho^{\Gamma}$ is positive-semidefinite the system is separable, otherwise it is entangled. In what follows a density matrix $\rho$ is said to be \textit{valid} iff it is positive-semidefinite.\\
{If the $15$ types of $X$-states are expressed in terms of Pauli operators (Eq. (\ref{fano})), we find that $6$ of them are endowed with $4$ $\tau$  coefficients and $3$ $\beta$ ones; these are states that belong to Group 1. The $9$ types of Group 2 entail states with $2$ $\tau$'s and $5$ $\beta$'s.}
 As already pointed out, the state given by Eq. (\ref{fano}) is from Group 2. \\
The eigenvalues $\lambda_{i}$ of density matrices of Group 1 states are found to be equal to those of their partial transposes $\lambda^{\Gamma}_{i}$,
\begin{equation}
    \lambda_{i} = \lambda^{\Gamma}_{i} = \left\{ 
    \begin{array}{l}
        \frac{1}{4}\left( \tau_{0} + 1 \pm \sqrt{(\beta_{1} + \tau_{1})^{2} + (\beta_{2} + \tau_{2})^{2} + (\beta_{3} + \tau_{3})^{2}} \right) \\
        \frac{1}{4}\left( - \tau_{0} + 1 \pm \sqrt{(\beta_{1} - \tau_{1})^{2} + (\beta_{2} - \tau_{2})^{2} + (\beta_{3} - \tau_{3})^{2}} \right),  
    \end{array}\right.
\end{equation}
where we have introduced generalised parameters for the 6 different states - the three $\beta_{i}$ parameters correspond to three correlation operators sharing a common tensor factor, the three $\tau_{i}, i = \{1,2,3\}$, correspond to the coordinates of the Bloch vector of one partially reduced state, and $\tau_{0}$  is the coordinate of the Bloch vector of the remaining partially reduced state. For example, $F_{[0:0:0:1]}$ maps to the following Group 1 state:
\begin{equation}
    \rho_{[0:0:0:1]} = \frac{1}{4}[I \otimes I + \tau^{A}_{x}X \otimes I + \tau^{A}_{y}Y \otimes I + \tau^{A}_{z}Z \otimes I + \tau^{B}_{z}I \otimes Z + \beta_{xz} X \otimes Z + \beta_{yz} Y \otimes Z + \beta_{zz} Z \otimes Z].
\end{equation}
As $\{\lambda_{i}\} = \{\lambda^{\Gamma}_{i}\}$, the positive-semidefinite criteria for $\rho$ and $\rho^{\Gamma}$ render all Group 1 states separable.\\
We can also introduce generalised parameters for Group 2 states:
\begin{equation}\label{G2_gen_params}
    \begin{array}{ll}
    (\tau_{1}, \tau_{2}) &= (\tau^{A}, \tau^{B}), \\
    \beta_{0} &:= \{\beta_{ij} | \beta_{ik} = 0 = \beta_{lj} \quad \forall k \neq i, l \neq j\}, \\
    M &= \begin{pmatrix} 
        \beta_{1} & \beta_{2} \\
        \beta_{3} & \beta_{4}
        \end{pmatrix} := \beta(i,j) \quad \text{for} \quad \beta_{0} \neq \beta_{ij},
\end{array}
\end{equation}
where $\tau_{i}$ are the coordinates of the Bloch vectors of the partially reduced states  (always one nonzero coordinate for each subsystem), $\beta_{0}$ is the unique $\beta$ parameter whose operator has no common factors with the other nontrivial operators, and $M$ is the submatrix formed from the $\beta$ matrix by removing the row and column containing $\beta_{0}$. For the Eq. \eqref{fano} example,
\begin{equation*}
\begin{array}{rl}\vspace{2pt}
    \beta &= \begin{pmatrix}
    \beta_{xx} & \beta_{xy} & \\
    \beta_{yx} & \beta_{yy} & \\
     & & \beta_{zz}
     \end{pmatrix},\\\vspace{2pt}
    M &= \begin{pmatrix}
    \beta_{xx} & \beta_{xy} \\
    \beta_{yx} & \beta_{yy}
    \end{pmatrix},\\
    \beta_{0} &= \beta_{zz}.
\end{array}
\end{equation*}\\
\\
In this construction, eigenvalues of the density matrix of a Group 2 state  have two general forms:\\
\textsc{Type I}:\\
\begin{equation}\label{eq:group2_type1_eigs}
\begin{array}{l}
    \lambda_{i\textsc{,I}} = \left\{ \begin{array}{l}
    \frac{1}{4} \left( \beta_{0} + 1 \pm \sqrt{(\beta_{1} - \beta_{4})^{2} + (\beta_{2} + \beta_{3})^{2} + (\tau_{1} + \tau_{2})^{2}} \right), \\
    \frac{1}{4} \left( -\beta_{0} + 1 \pm \sqrt{(\beta_{1} + \beta_{4})^{2} + (\beta_{2} - \beta_{3})^{2} + (\tau_{1} - \tau_{2})^{2}} \right),\end{array} \right. \\
    \lambda^{\Gamma}_{i\textsc{,I}} = \left\{ \begin{array}{l}\frac{1}{4} \left( \beta_{0} + 1 \pm \sqrt{(\beta_{1} + \beta_{4})^{2} + (\beta_{2} - \beta_{3})^{2} + (\tau_{1} + \tau_{2})^{2}} \right), \\
    \frac{1}{4} \left( -\beta_{0} + 1 \pm \sqrt{(\beta_{1} - \beta_{4})^{2} + (\beta_{2} + \beta_{3})^{2} + (\tau_{1} - \tau_{2})^{2}} \right). \end{array}\right. 
\end{array}
\end{equation} 
Type II is obtained by the switch $(\lambda_{i} \leftrightarrow \lambda^{\Gamma}_{i})$, or, equivalently, by the transformations $(\beta_{0} \mapsto -\beta_{0}), (\tau_{1(2)} \mapsto -\tau_{1(2)})$.\\
These forms allow for entangled states, and we will later show the corresponding conditions for states with maximally mixed subsystems. We also assign to each Group 2 $X$-state a parameter $t$ to denote its type; $t=1$ for Type I states, and $t=2$ for Type II ones.\\
The parameters in $\lambda, \lambda^{\Gamma}$ above can be expressed as coordinates in $\mathbf{R}^{2}$, with the radical term giving the Euclidean distance between pairs of points. This handy representation will allow us to find regions of validity, separability and entanglement, using the following proposition:
\begin{prop}\label{prop:beta_0_L_sign}
Let $L^{+}:= \sqrt{(\beta_{1} + \beta_{4})^{2} + (\beta_{2} - \beta_{3})^{2}}$ and $L^{-}:= \sqrt{(\beta_{1} - \beta_{4})^{2} + (\beta_{2} + \beta_{3})^{2}}$. For a valid, entangled generalised X-state $\rho$ with $\tau_{1} = \tau_{2}=0$ and eigenvalues $\lambda_{i\textsc{,I}} = \{1 \pm \beta_{0} (\pm) L^{\mp}\}$, $\beta_{0} < 0 \iff L^{+} > L^{-}$.
\end{prop}
\begin{proof}
($\implies$): The validity constraint yields
\begin{equation}\label{eq:group_2_pos_def_L}
\begin{array}{llll}
1-\beta_{0} &\geq L^{+} &\qquad\text{ and } \qquad1+\beta_{0} &\geq L^{-}.
\end{array}
\end{equation}
For $\beta_{0} < 0$ and min$\{\lambda^{\Gamma}_{i\textsc{,I}}\}<0$, 
\begin{equation*}
\begin{array}{llll}
    1-|\beta_{0}| &< L^{+} \leq 1+|\beta_{0}| \qquad\text{ or } \qquad &1+|\beta_{0}| &< L^{-} \leq 1-|\beta_{0}|.
    \end{array}
\end{equation*}
Only the left-hand-side inequalities are consistent, and can be combined with the validity conditions to give $L^{-} \leq 1-|\beta_{0}| < L^{+}$. \\
($\Leftarrow$): We have
\begin{equation*}
    1-\beta_{0} \geq L^{+} > L^{-},
\end{equation*}
\begin{equation*}
    \implies 1+\beta_{0} < L^{+} \leq 1-\beta_{0} \qquad\text{ or } \qquad 1-\beta_{0} < L^{-} \leq 1 - \beta_{0},
\end{equation*}
which again invalidates the right-hand-side inequalities, giving
\begin{equation*}
    -\beta_{0} > \beta_{0} \implies \beta_{0} < 0.
\end{equation*}
\end{proof}
\noindent
Now, we can simply take $\beta_{0} < 0$ and the PPT criterion acquires the form
\begin{equation}\label{eq:group_2_PPT_L_plus}
1 - |\beta_{0}| < L^{+}.
\end{equation}
Taking $\beta_{0}>0$ yields
\begin{equation}\label{eq:group_2_PPT_L_minus}
1 - |\beta_{0}| < L^{-}.
\end{equation}
This thus allows us to consider the $\beta_{0}<0$ case to determine the regions of validity, separability and entanglement, and then make suitable transformations for the $\beta_{0}>0$ case.
\begin{prop}
\label{prop:regions_sep_ent}
Let $C:= (\beta_{4},\beta_{3}), D:=(\beta_{1},\beta_{2}), E:= (\beta_{1},-\beta_{2}), F:=(\beta_{4},-\beta_{3}) \in \mathbf{R}^{2}$, $r:= 1 - |\beta_{0}|$, $R:=1+|\beta_{0}|$, and $(C,r)$ be the closed disc centered at $C$ with radius $r$, etc.\\
For a given Group 2 $X$-state with maximally-mixed subsystems represented by $\rho$ with given $\beta_{0}, \beta_{3}, \beta_{4}$, its region of validity is given by $\mathcal{V} := (C,r)\cap(-C,R)$, the region of separability by $\mathcal{S}:=(C,r)\cap(-C,r)$, and the region of entanglement by $\mathcal{E}:=\mathcal{V}\setminus\mathcal{S}$. The matrix $\rho$ is valid iff $E \in  (-1)^{t}\sgn(\beta_{0})\mathcal{V}$, separable for $E \in \mathcal{S}$, and entangled for $E \in (-1)^{t}\sgn(\beta_{0})\mathcal{E}$. The following conditions hold: 
\begin{enumerate}
    \item $\mathcal{V} \neq \emptyset \iff \beta_{3}^{2} + \beta_{4}^{2} \leq 1$;
    \item $\mathcal{S} \neq \emptyset \iff \beta_{3}^{2} + \beta_{4}^{2} \leq (1 - |\beta_{0}|)^{2}$.
\end{enumerate}
Equivalently, for a given generalised Group 2 X-state $\rho$ with given $\beta_{0}, \beta_{1}, \beta_{2}$, its region of validity for $F$ is given by $\mathcal{V'} := (E,r)\cap(-E,R)$, the region of separability by $\mathcal{S'}:=(E,r)\cap(-E,r)$, and the region of entanglement by $\mathcal{E'}:=\mathcal{V'}\setminus\mathcal{S'}$ with corresponding conditions: 
\begin{enumerate}
    \item $\mathcal{V'} \neq \emptyset \iff \beta_{1}^{2} + \beta_{2}^{2} \leq 1$;
    \item $\mathcal{S'} \neq \emptyset \iff \beta_{1}^{2} + \beta_{2}^{2} \leq (1 - |\beta_{0}|)^{2}$.
\end{enumerate}
$\rho$ is valid iff $F \in (-1)^{t}\sgn(\beta_{0})\mathcal{V'}$, separable for $F \in \mathcal{S'}$, and entangled for $F \in (-1)^{t}\sgn(\beta_{0})\mathcal{E'}$.
\end{prop}
\begin{proof}
We consider a state of Type \textsc{I} with maximally mixed subsystems, with $\beta_{0}<0$, and prove the first half of the proposition. The second half follows readily from the same chain of arguments. For validity, by Proposition \ref{prop:beta_0_L_sign}, we have
\begin{subequations}\label{eq:validity_type1_beta_less_0}
\begin{equation}
    L^{+}\leq1+|\beta_{0}| \label{eq:validity_L_plus},
\end{equation}
\begin{equation}
    L^{-}\leq1-|\beta_{0}| \label{eq:validity_L_minus}.
\end{equation}
\end{subequations}
The quantity $L^{-}$ ($L^{+}$) is given by the Euclidean distance between points $C$ and $E$ ($C$ and $-E$). Via the representation  of the parameters in $\mathbf{R}^{2}$ (see Figure \ref{fig:validity_separability_graph_colours_main}), the condition (\ref{eq:validity_L_minus}) means that the point $E$ is contained within the closed disc $(C,r)$ and, equally, that the point $-E$ lies in $(-C,r)$. Condition (\ref{eq:validity_L_plus}) means that the point $-E$ is contained within $(C,R)$. Thus, for both to hold, $-E$ must lie within $(-C,r)\cap(C,R)$. Reflecting through the origin, this implies that $E \in \mathcal{V}:=(C,r)\cap(-C,R)$. This region is nonempty when the points $C, -C$ are closer to each other than $r+R$, i.e. when $|C|^{2} = \beta_{4}^{2}+\beta_{3}^{2} \leq 1$. For separability, we need that min$\{\lambda^{\Gamma}_{i\textit{,I}}\}\geq0$, i.e. that $L^{+}\leq1-|\beta_{0}|$. This gives $-E \in (C,r)$, or, equivalently, that $E \in (-C,r)$. For validity and separability that requires $E \in \mathcal{S}:=(C,r)\cap(-C,r)$, which when reflected through the origin also gives $-E\in \mathcal{S}$. This region is nonempty when $|C|^{2}\leq r^{2}=(1-|\beta_{0}|)^{2}$. For entanglement that occurs when the state is valid but not separable we obtain $E \in \mathcal{V} \backslash \mathcal{S}$. We note now that for the case $\beta_{0}>0$, via  proposition \ref{prop:beta_0_L_sign}, the conditions (\ref{eq:validity_type1_beta_less_0}) exchange $L^{+} \leftrightarrow L^{-}$, or, equivalently, $E \mapsto -E$. This does not affect the separability condition as $\mathcal{S}$ is symmetric under these transformations, but it does change the sign of the regions $\mathcal{V}, \mathcal{E}$ that $E$ must belong to. Thus, by multiplying by the sign of $\beta_{0}$ we can cover both the scenarios. Finally, we note that under the transformation $\beta_{0} \mapsto -\beta_{0}$ we can infer the conditions for Type \textsc{II}. As this sign change is equivalent to exchanging $L^{+} \leftrightarrow L^{-}$ in proposition (\ref{prop:beta_0_L_sign}), it means that we must once again multiply by a factor of $(-1)$ to exchange types. Then we see that $E \in (-1)^{t}\sgn(\beta_{0})\mathcal{V}$ covers our initial case and any of the mentioned transformations.
\end{proof}
{\begin{remark}
 Proposition \ref{prop:regions_sep_ent} provides a constructive way of generating examples of generalised $X$-states with maximally-mixed subsystems that are entangled as shown in Figures \ref{fig:Validity_separability_graph} and \ref{fig:validity_separability_graph_colours_main}.
\end{remark}}
\begin{figure}[!h] \centering
    \centering
    \includegraphics[height=7cm]{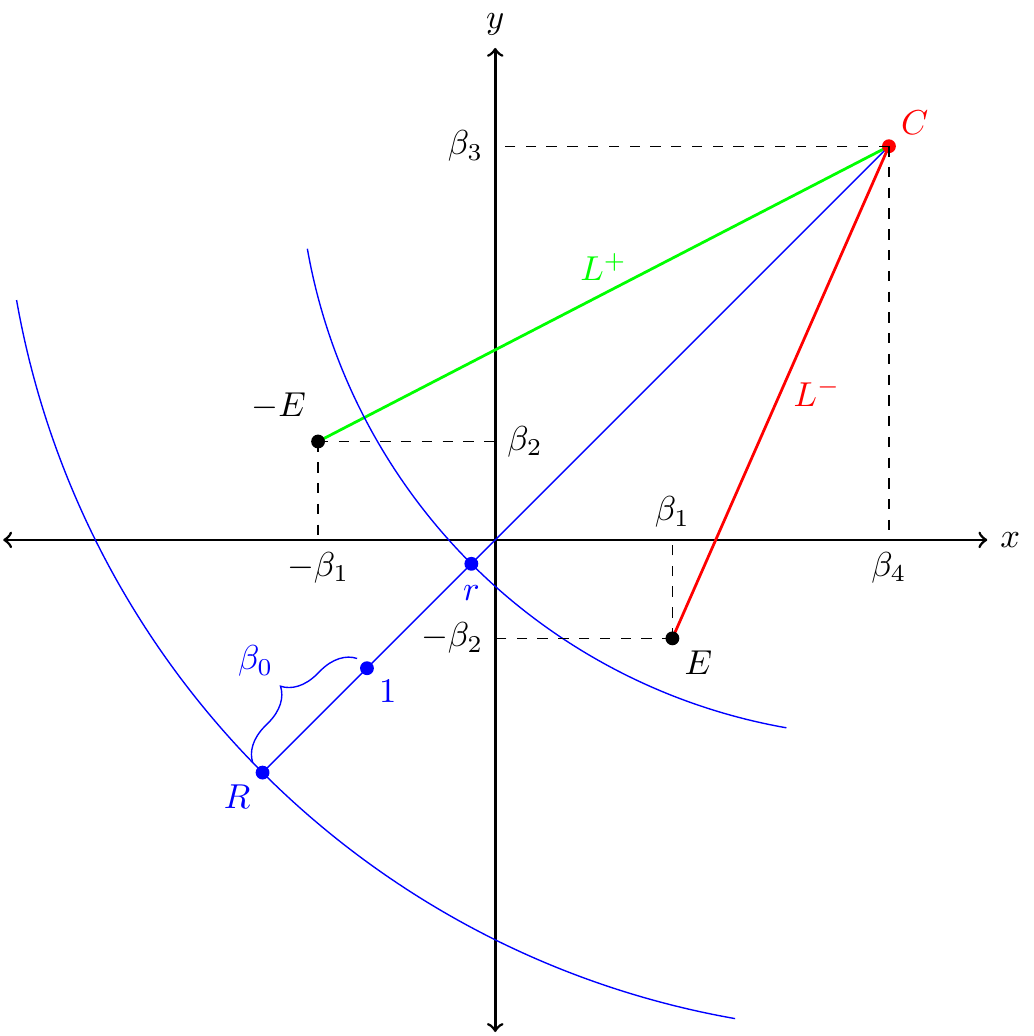}
    \caption{The parameter space for a positive-definite Group 2 state (i.e., a state with maximally-mixed subsystems and of Type I, $\beta_{0}<0$). The parameters $\beta_{1}, \beta_{4}$ are plotted on the $x$-axis, $\beta_{2},\beta_{3}$ on the $y$-axis. As $E \in (C,r) \land -E \in (C,R)$, this state is valid as per proposition \ref{prop:regions_sep_ent}. Moreover, as $L^{+} > r = 1 - |\beta_{0}|$, the state is also entangled.}
    \label{fig:Validity_separability_graph}
\end{figure}
\vspace*{0.5cm}
\begin{figure}[h!]
   \centering
\includegraphics[height=7cm]{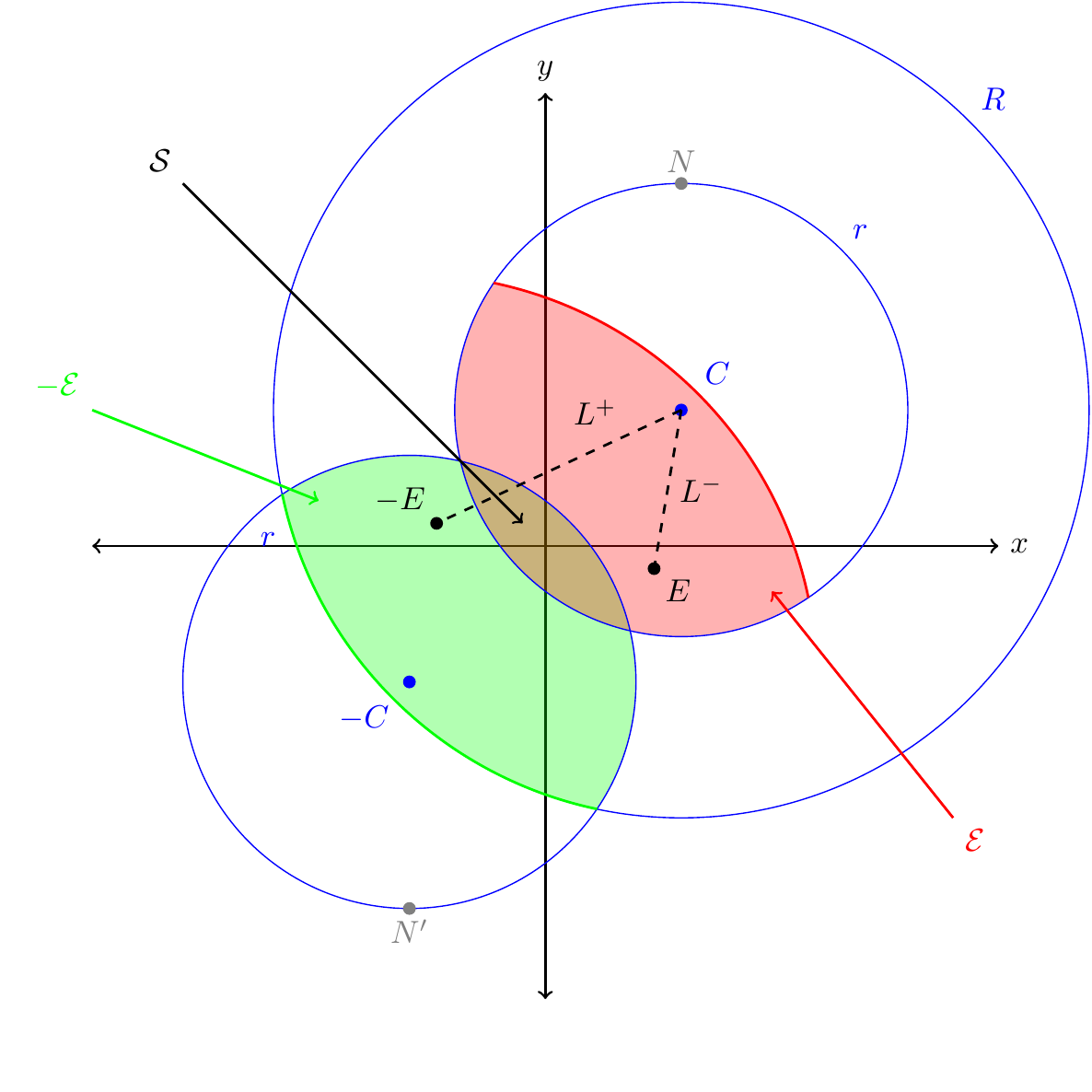}
   \caption{The regions of validity, separability and entanglement for a (maximally-mixed) Group 2 $X$-state, in the same coordinate system as employed in Figure \ref{fig:Validity_separability_graph}. The points of the  closed disc $(-C,r)$ are those of $(C,r)$ when reflected through the origin, as illustrated by the labels $N, N'$. The regions of entanglement $\mathcal{E}$ and separability $\mathcal{S}=(C,r)\cap(-C,r)$ are shown, with the region of validity $\mathcal{V} = (C,r)\cap(-C,R) = \mathcal{E} \oplus \mathcal{S}$. The state shown is entangled for $(-1)^{t}\sgn({\beta_{0}})E \in \mathcal{E}$.}
    \label{fig:validity_separability_graph_colours_main}
\end{figure}

\section{$X$-states as geometric hyperplanes of $\mathcal{W}(3,2)$}
Let us now focus on the symplectic polar space of order two and rank two, $\mathcal{W}(3,2)$ , i.e. the space of all totally isotropic subspaces of PG$(3,2)$ with respect to a given symplectic form. The space $\mathcal{W}(3,2)$ encodes geometrically the commutation relations between the elements  of $\mathcal{P}_2$ in the following sense (see \cite{saniga2007multiple}). Given the symplectic form $\sigma$ of Eq.(\ref{sigma}),  a totally isotropic line $(p,q,r)$ of PG$(3,2)$ is a line such that $\sigma(p,q)=\sigma(p,r)=\sigma(q,r)=0$. Then any representatives  $\mathcal{O}_p$, $\mathcal{O}_q$ and $\mathcal{O}_r$ of the classes mapped to $p,q$ and $r$ such that $(p,q,r$) is a totally isotropic line, represent a triple of mutually commuting observables. 

$\mathcal{W}(3,2)$ can also be viewed as a point-line incidence structure $\mathcal{G}=(\mathcal{P},\mathcal{L},\mathcal{I})$, where $\mathcal{P}$ are the $15$ points and $\mathcal{L}$ are the $15$ totally isotropic lines  of PG$(3,2)$ , $\mathcal{I}\subset \mathcal{P}\times \mathcal{L}$ being the incidence relation, i.e. a set-theoretic inclusion of points in lines. The point-line geometry corresponding to $\mathcal{W}(3,2)$ is a unique triangle-free $15_3$-configuration ($15$ points/lines, $3$ points per line and $3$ lines through a point) known as the Doily, or the Cremona-Richmond configuration.
Restricting to canonical representatives of the classes of $\mathcal{P}_2$ (i.e. $s=1$ in Eq. (\ref{eq:pi})), one obtains one of the `standard' parametrizations of the Doily as illustrated in Figure \ref{thedoily}.
\begin{figure}[!h] \centering
\centering
 \includegraphics[width=6cm]{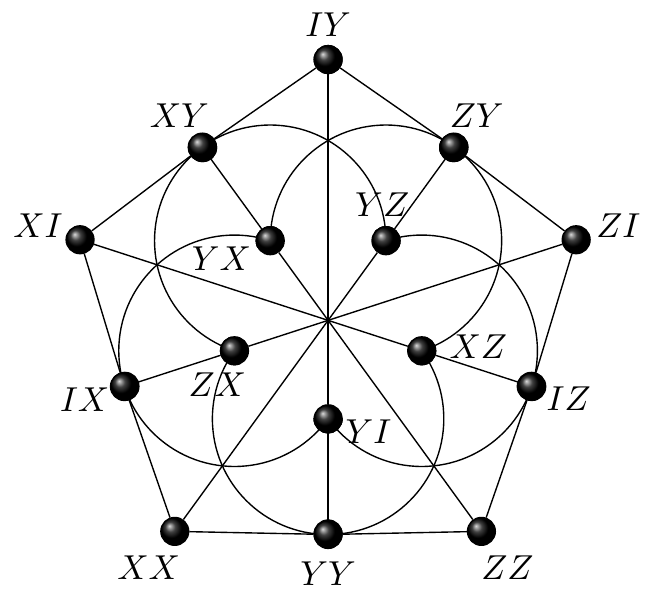}
 \caption{The Doily with its points labeled by two-qubit observables. Two observables commute if they are collinear. Note that lines are represented not only by straight segments, but also by arcs of circles.}\label{thedoily}
\end{figure}

Labeled Fano planes were defined as sets of observables commuting with a given observable. The Fano plane of Figure \ref{fanoplane} represents the set of observables commuting with $ZZ$; geometrically, it is the set of points $q$ of PG$(3,2)$ such that $\sigma(p,q)=0$ for $p=[1:0:1:0]$. 
In our labeled Doily the trace of the labeled Fano plane corresponds to three concurrent lines; the corresponding set of points is called a perp-set of the point of concurrence. 
The perp-set of the point $ZZ$ is illustrated in Figure \ref{zzperp}.  
\begin{figure}[!h] \centering
\includegraphics[width=6cm]{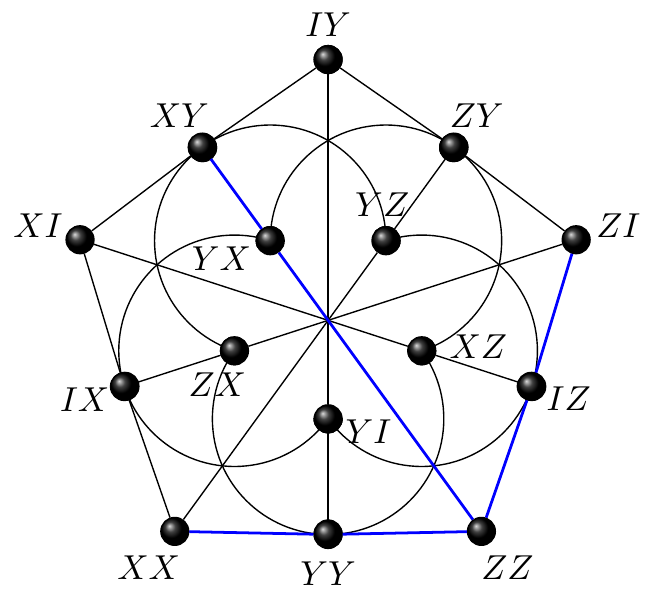}
 \caption{The perp-set (blue) defined by $ZZ$, which corresponds to the trace on $\mathcal{W}(3,2)$ of the Fano plane $F_{[1:0:1:0]}$ of PG$(3,2)$.}\label{zzperp}
\end{figure}
A fact of crucial importance for us is that perp-sets of $\mathcal{W}(3,2)$ are also geometric hyperplanes of the configuration (Definition \ref{hyperplane}).
\begin{definition}\label{hyperplane}
 Let $\mathcal{G}=(\mathcal{P},\mathcal{L},\mathcal{I})$ be a point-line incidence structure. A geometric hyperplane $H$ of $\mathcal{G}$ is a subset of $\mathcal{P}$ such that a line of $\mathcal{L}$ is either contained in $H$, or has just a single point in common with $H$.
\end{definition}

Consider now the hyperbolic quadric of $\mathcal{W}(3,2)$ defined as:
\begin{equation}\label{eq:quadric}
 \mathcal{Q}_0=\{p=[x_1:x_2:x_3:x_4]\in \mathcal{W}(3,2), x_1x_2+x_3x_4+x_1+x_2+x_3+x_4=0\}.
\end{equation}
It is the unique quadric of $\mathcal{W}(3,2)$ involving only non-trivial Pauli matrices. This quadric is illustrated in red in Figure \ref{quadric}. All quadrics of $\mathcal{W}(3,2)$ are also  geometric hyperplanes in  the sense of Definition  \ref{hyperplane}.  

\begin{figure}[!h] \centering
 \includegraphics[width=6cm]{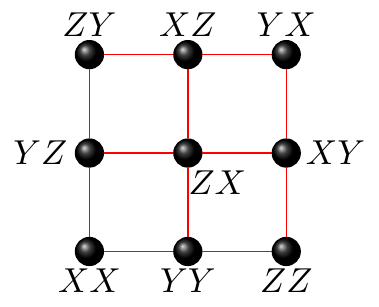}
 \includegraphics[width=6cm]{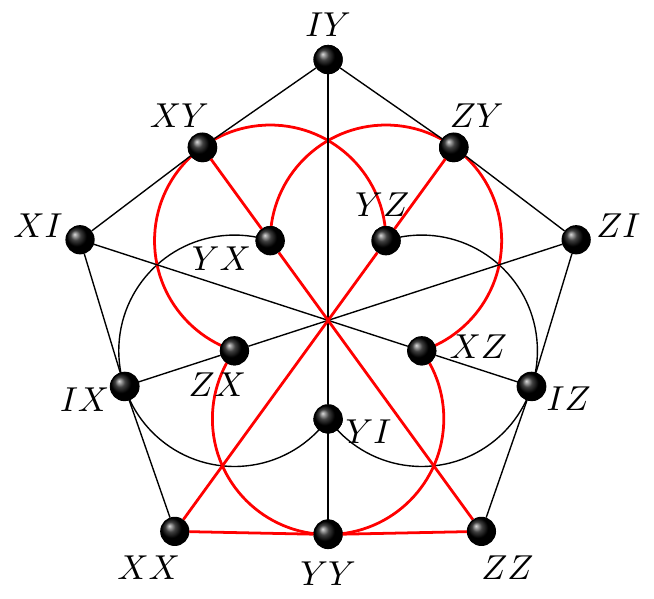}
 \caption{A distinguished hyperbolic quadric in $\mathcal{W}(3,2)$. Left: Each hyperbolic quadric in  $\mathcal{W}(3,2)$ has the point-line structure of a grid; the quadric in question involves only non-trivial Pauli matrices. Right: The same quadric (red) viewed as a geometric hyperplane of $\mathcal{W}(3,2)$.}\label{quadric}
\end{figure}

Definition  \ref{hyperplane} implies that if $H_1$ and $H_2$ are two distinct geometric hyperplanes, their intersection $H_1\cap H_2$ is a geometric hyperplane of the subgeometries defined by  $H_1$ and $H_2$. The quadric $\mathcal{Q}_0$ whose point-line structure can be pictured as a grid (Figure \ref{quadric}) has only two types of geometric hyperplanes, perp-sets and ovoids. Thus the intersections of the $15$ perp-sets of $\mathcal{W}(3,2)$ with $\mathcal{Q}_0$ will be of two different kinds: transverse intersections, corresponding to ovoids (Figure \ref{intersect}) or tangential intersections, corresponding to perp-sets  (Figure \ref{intersect2}).

\begin{figure}[!h] \centering
 \includegraphics[width=6cm]{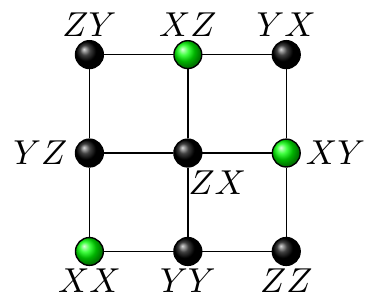}
 \includegraphics[width=6cm]{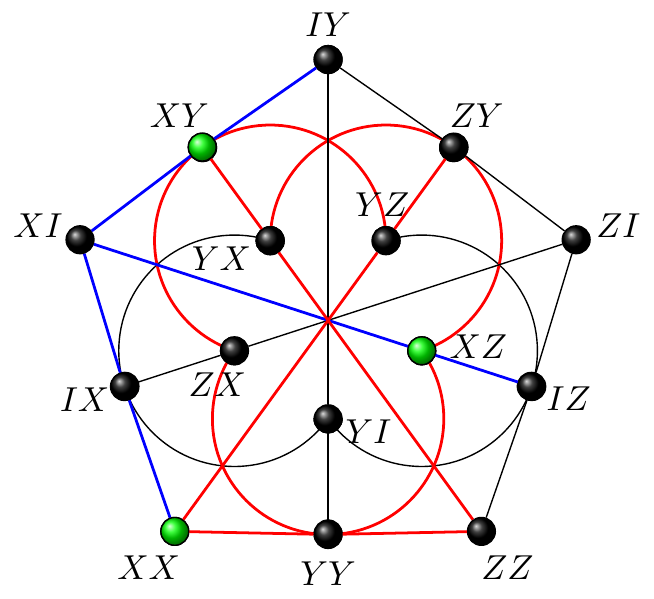}
 \caption{The perp-set $H_{XI}$ of $\mathcal{W}(3,2)$ intersect $\mathcal{Q}_0$ transversally. Left: $\mathcal{Q}_0\cap H_{XI}$ (green) is a geometric hyperplane called an ovoid (three points, no two on a line). Right: The same intersection portrayed in $\mathcal{W}(3,2)$.}\label{intersect}
\end{figure}
\begin{figure}[!h] \centering
 \includegraphics[width=6cm]{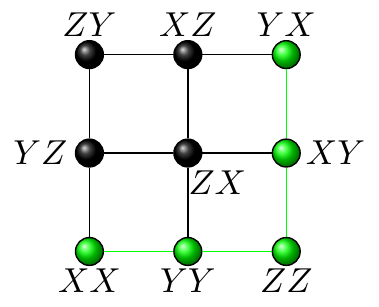}
 \includegraphics[width=6cm]{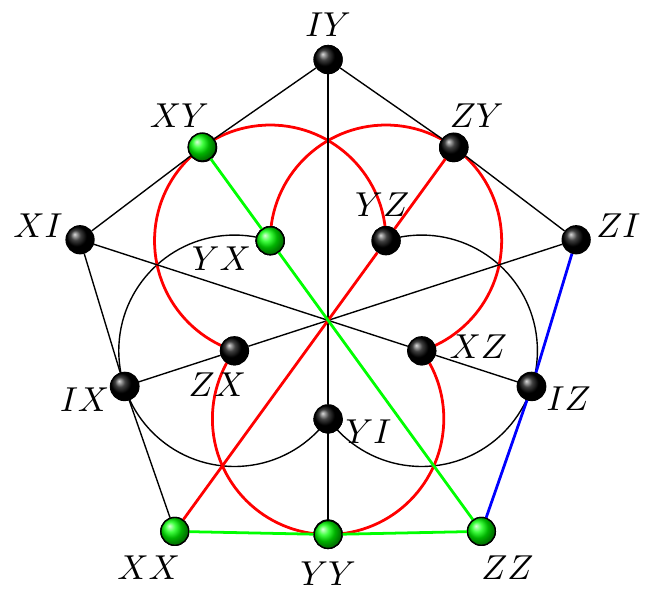}
 \caption{A tangential intersection of $\mathcal{Q}_0$ with a perp-set ($H_{ZZ}$) of $\mathcal{W}(3,2)$.}\label{intersect2}
\end{figure}

We thus come to a very important observation that furnishes a geometric interpretation of the $9+6$ splitting of the $X$-states described in Section \ref{entanglement}.
\begin{prop}
 The $9$ types of $X$-states of Group 2  correspond to the $9$ perp-sets of $\mathcal{W}(3,2)$ that intersect $\mathcal{Q}_0$ tangentially, whereas the $6$ $X$-states of Group 1 correspond to the $6$ perp-sets of $\mathcal{W}(3,2)$ that intersect $\mathcal{Q}_0$ transversally.
\end{prop}

\proof The result follows by considering the $15$ types of $X$-states given in Appendix \ref{appA}. $\Box$


\section{Hyperplane-states in the two-qubit Pauli group and states with maximally-mixed subsystems}
In the previous section we have shown how the $15$ distinct sets of two-qubit Pauli observables that generate two-qubit $X$-states correspond to perp-sets, a specific class of geometric hyperplanes, in $\mathcal{W}(3,2)$ whose points are parametrized by observables as shown in Figure \ref{thedoily}. This leads naturally to considering other geometric hyperplanes of $\mathcal{W}(3,2)$ and the sets of observables lying on them for generating  density matrices.  Let us make this idea more precise with the following definition:
\begin{definition}
 Let $\rho$ be a two-qubit density matrix. One says that $\rho$ is a two-qubit hyperplane-state iff the set of two-qubit observables defining $\rho$ is a geometric hyperplane of $\mathcal{W}(3,2)$.
\end{definition}

\begin{example}
 $X$-states are thus a special type of hyperplane-states as they correspond to perp-sets. Let us now consider a density matrix of the form given by Eq. (\ref{mixed}).
 \begin{equation}\label{mixed}
  \rho=\frac{1}{4}(I_4+\beta_{xx}XX+\beta_{xy}XY+\beta_{xz}XZ+\beta_{yx}YX+\beta_{yy}YY+\beta_{yz}YZ+\beta_{xz}XZ+\beta_{yz}YZ+\beta_{zz}ZZ).
 \end{equation}
It corresponds to the generic form of  a maximally-mixed two-qubit state. In terms of hyperplane-state description, this density matrix  is generated by the quadric $\mathcal{Q}_0$.
\end{example}

At this point we need to introduce some more details about geometric hyperplanes of $\mathcal{W}(3,2)$ \cite{Saniga2007}. There are altogether three kinds of them, namely:

\begin{itemize}
 \item $15$ perp-sets $H_p$, defined for each point $p\in \mathcal{W}(3,2)$ as $H_p=\{q \in \mathcal{W}(3,2), \sigma(p,q)=0\}$. We have met several examples of them; they  correspond to the $15$ types of $X$-states.
 \item $10$ grids, or Mermin-hyperplanes. Each of them comprise $9$ points and $6$ lines, with three points per line and two lines through a point. The quadric $\mathcal{Q}_0$ given by Eq. (\ref{eq:quadric}) and portrayed in Figure \ref{quadric} serves as an illustrative  example. If we take this grid embedded in $\mathcal{W}(3,2)$ and rotate it by $2\pi/5$ degrees around the center of the figure, one obtains $4$ more grids  ($\mathcal{Q}_1,\dots,\mathcal{Q}_4$). Figure \ref{fig:merminovoid} shows the second form of a grid embedded in $\mathcal{W}(3,2)$, referred to as $\mathcal{Q}_5$. Performing the same rotation as in the previous case yields the four remaining grids $\mathcal{Q}_6,\dots,\mathcal{Q}_9$. One may call such hyperplanes Mermin-hyperplanes, as each of them furnishes an observable-based proof (a Mermin-Peres `magic' square) of the famous Kochen-Specker theorem (see \cite{mermin1993hidden} and \cite{peres1991two} for the original argument and \cite{Saniga2007,planat2007pauli,holweck2017contextuality} for the discussions of the geometrical contexts).
 \item $6$ ovoids. An ovoid of $\mathcal{W}(3,2)$ is a set of five points, no two of them being collinear. Figure \ref{fig:merminovoid} depicts two ovoids in $\mathcal{W}(3,2)$. The first one, $\mathcal{O}_1$, is rotationally invariant. Rotating the second one, $\mathcal{O}_2$, by $2\pi/5$ degrees one obtains the remaining $4$ ovoids $\mathcal{O}_i$, $i=3,\dots,6$. Like Mermin-hyperplanes, ovoids also underlie a certain family of quantum contextual configurations (namely three-qubit Mermin pentagrams) using the Klein correspondence \cite{saniga2012mermin}.
\end{itemize}

\begin{figure}[!h] \centering
\begin{center}
 \includegraphics[width=5cm]{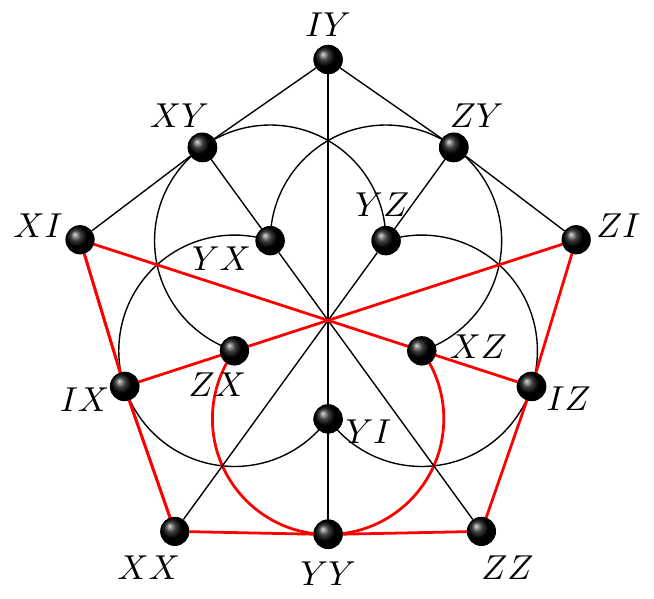}
 \includegraphics[width=5cm]{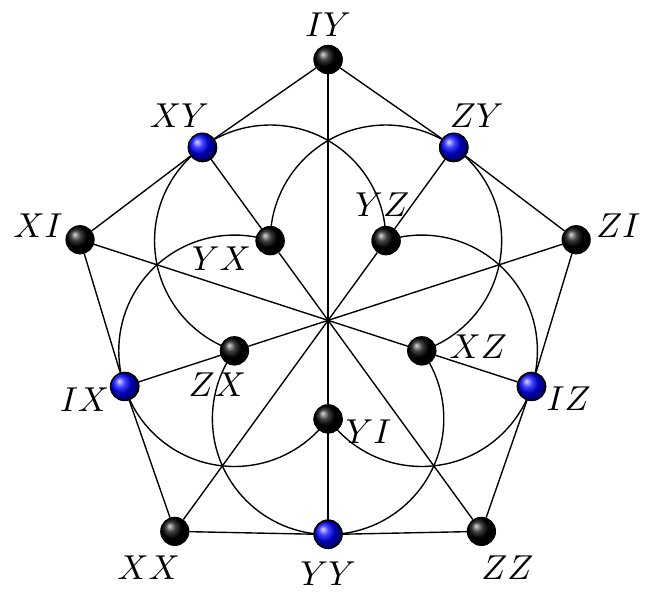}
 \includegraphics[width=5cm]{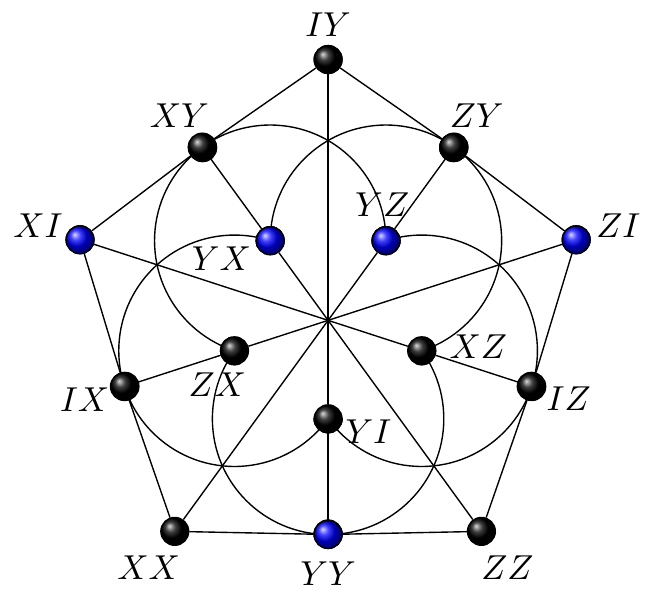}
 \caption{Left: The Mermin-hyperplane denoted as $\mathcal{Q}_5$. Its rotation by $2\pi/5$ degrees around the center of the Doily yields four more grids. Middle and Right: The ovoids of $\mathcal{W}(3,2)$. Rotation of the second one by $2\pi/5$ degrees gives four more ovoids.}\label{fig:merminovoid}
 \end{center}
\end{figure}
To each of the $31$ geometric hyperplanes listed above one can associate a specific type of two-qubit density matrix. For instance, if one consider the hyperplanes $\mathcal{Q}_5$ and $\mathcal{O}_1$, one gets the following types of quantum states
\begin{equation}
 \rho_{\mathcal{Q}_5}=\frac{1}{4}(I_4+\tau_{x}^A XI+\tau^A _z ZI+\tau^B _x IX+\tau^B _z IZ+\beta_{xx}XX+\beta_{YY}YY+\beta_{zz}ZZ+\beta_{zx}ZX+\beta_{xz}XZ)
\end{equation}
and 
\begin{equation}
 \rho_{\mathcal{O}_1}=\frac{1}{4}(I_4+\tau^B _x IX+\tau^B _z IZ+\beta_{xy} XY+\beta_{zy} ZY+\beta_{yy}YY),
\end{equation}
respectively.

The $10$ Mermin-hyperplane states depend generally on $9$ parameters, whereas the $6$ ovoid-states on $5$ parameters. However, if we restrict to states with maximally-mixed subsystems, i.e. states such that their partially reduced states are maximally mixed, $\rho_A=\frac{1}{2}I_2$ and $\rho_B=\frac{1}{2}I_2$, or, in our context, two-qubit density matrices such that the $\tau$ coefficients are zero, we arrive at the following result:
\begin{prop}
 Let us consider the $16$-member family of hyperplane-states that is the union of $10$ Mermin-hyperplane-states and $6$ ovoid-states. Then
 \begin{itemize}
  \item The $\mathcal{Q}_0$-hyperplane states correspond to general two-qubit  states with maximally-mixed subsystems.
  \item The $9$ $\mathcal{Q}_i$-hyperplane states  correspond to  $X$-states of Group 2 with maximally-mixed subsystems.
  \item The $6$ $\mathcal{O}_i$-hyperplane states correspond to  $X$-states of Group 1 with maximally-mixed subsystems.
 \end{itemize}
\end{prop}
\proof A case by case argument leads to the result. $\Box$
\section{Non-locality of hyperplane-states}
To address non-local properties of hyperplane-states, we will borrow the following theorem from Horodecki et. al. \cite{horodecki1995violating}:
\begin{theorem}[Horodecki R., Horodecki P., Horodecki M. \cite{horodecki1995violating}]
A density matrix $\rho$ describes a state that violates the Bell inequality iff $\mathcal{M}>1$, where $\mathcal{M}=u + \bar{u}$, the sum of the two largest eigenvalues of the matrix $\beta^{T}\beta$.
\end{theorem}
For both Group 1 $X$-states and ovoid-states, their consistent separability renders them locally realistic for all valid choices of parameters.
Concerning Group 2 and Mermin-states $\mathcal{Q}_1,\dots,\mathcal{Q}_9$, we will examine those that share the same $\beta$ matrices. (As already mentioned, the remaining Mermin-state, $\mathcal{Q}_{0}$, corresponds to the most general $2$-qubit states with maximally-mixed subsystems).
The eigenvalues $u_{i}$ of $\beta^{T}\beta$ can be expressed in terms of generalised parameters \eqref{G2_gen_params} as follows
\begin{equation}
    u_{i} = \{\beta_{0}^{2}, \frac{1}{2}(B(M) \pm U(M))\},
\end{equation}
where $B(M):=\text{Tr}M^{T}M$ and $U:=\sqrt{B^{2}-4(\text{det}M)^{2}}$. The two eigenvalues $\frac{1}{2}(B \pm U)$ are the eigenvalues of the matrix $M^{T}M$, which we denote by $m_1, m_2$ with $m_1 = \frac{1}{2}(B-U)$, and $m_2 = \frac{1}{2}(B+U)$.\\
\\
One has then in terms of $B$ and $U$ that
\begin{equation}\label{eq:nonlocality}
    \mathcal{M} = \left\{\begin{array}{cc}\vspace{.2cm}
    B, & \beta_{0}^{2} < m_{1}, \\
    \beta_{0}^{2} + \frac{1}{2}(B+U), & \beta_{0}^{2} \geq m_{1}.
    \end{array}\right.
\end{equation}
This leads to the following proposition:
\begin{prop}
The set of (not necessarily valid) Group 2 X-states or Mermin-states $\{\rho\}$ with given $\beta_{0},\beta_{3},\beta_{4}$ and constant $\mathcal{M}$ is given by a connected subset of points on the union of a circle and ellipse centered at the origin, consisting of arcs of the circle within the ellipse, and arcs of the ellipse within the circle (see Fig. \ref{fig:constant_nonlocality}).
\end{prop}
\begin{proof}
Without loss of generality, one can choose generalised coordinates rotated such that $\beta_{3}=0$ and the points $C, -C$ lie on the $x$-axis.\\
As $m_1, m_2$ are eigenvalues of the matrix $M^{T}M$, one can write the characteristic polynomial of this matrix in generalised coordinates:
\begin{equation*}
    f(\lambda) = \lambda^{2} - \lambda \text{Tr}M^{T}M + \text{det}M^{T}M
\end{equation*}
\begin{equation*}
    = \lambda^{2} - \lambda (\beta_{1}^{2} + \beta_{2}^{2} + \beta_{4}^{2}) + \beta_{1}^{2}\beta_{4}^{2}.
\end{equation*}
Taking an eigenvalue $m \in \{m_1, m_2\}$ as input, this results in 
\begin{equation*}
    0 = m^{2} - m (\beta_{1}^{2} + \beta_{2}^{2} + \beta_{4}^{2}) + \beta_{1}^{2}\beta_{4}^{2},
\end{equation*}
\begin{equation}\label{eq:ellipse}
    \Rightarrow \frac{\beta_{1}^{2}}{m} + \frac{\beta_{2}^{2}}{m - \beta_{4}^{2}} = 1.
\end{equation}
 For the denominators in \eqref{eq:ellipse} constant one obtains the equations for two conic sections, one for each $m_i$, $i=1,2$. For the denominators both positive, one has the equation for an ellipse, and for the first denominator positive and the second negative, one has the equation for a hyperbola. As $m_i$ are eigenvalues of the square matrix $M^{T}M$, they are necessarily non-negative, so it suffices to check the sign of $m_i - \beta_{4}^{2}$ for $m_i \neq 0$.\\
 First, one must check under what conditions one can retrieve constant denominator terms. For $\mathcal{M}=k$ constant, one has that $\text{max}\{m_1 + m_2, \beta_{0}^{2} +m_2\}$ is constant. Taking each argument individually, $m_1 + m_2 = B = k$ gives (recalling $B = \sum_{i=1}^{4} \beta_{i}^{2}$):
 \begin{equation*}
     \beta_{1}^{2} + \beta_{2}^{2} = k - (\beta_{3}^{2}+\beta_{4}^{2}),
 \end{equation*}
 i.e. the equation of a circle in coordinates with varying $\beta_{1},\beta_{2}$ and given $\beta_{3},\beta_{4}$, provided $k > \beta_{4}^{2} + \beta_{3}^{2}$. For $k< \beta_{3}^{2}+\beta_{4}^{2}$, no such solutions exist, and $\mathcal{M}$ can only be given by $\beta_{0}^{2} +m_2$. One has from proposition \ref{prop:regions_sep_ent} that a state can be valid only if $\beta_{3}^{2} + \beta_{4}^{2} \leq 1$, so a minimally valid system ($\beta_{3}^{2}+\beta_{4}^{2}=1$) allows only $B \geq 1$.\\
 \\
 Taking the other argument, $\beta_{0}^{2} + m_2$ constant for given $\beta_{0}$ indicates that $m_2$ is constant, which provides constant denominators in \eqref{eq:ellipse}. One can recover that $m_2 - \beta_{4}^{2} \geq 0$ by assuming the negation, and substituting the $\beta_{i}$ terms into $m_2$, with $\beta_{3}=0$. One then obtains $(\beta_{2}\beta_{4})^{2} < 0$, a contradiction. This then provides via \eqref{eq:ellipse} the equation for an ellipse, with semi-major and semi-minor axes given by $a=\sqrt{m_2}$ and $b=\sqrt{m_2-\beta_{4}^{2}}$, respectively, and foci at $\pm C = (\pm \beta_{4},0)$.\\
 (Note that the term $m_1-\beta_{4}^{2}$ can be similarly examined and found to be non-positive, however the constancy conditions on $\mathcal{M}$ do not give a constant $m_1$ term so the hyperbola given by substituting $m_1$ into \eqref{eq:ellipse} does not arise in this examination).\\
 \\
 The curves of constant $m_1+m_2$, $\beta_{0}^{2}+m_2$ are then given by a circle and an ellipse, respectively (see Fig. \ref{fig:constant_nonlocality}). The curve of constant $\mathcal{M} = \text{max}\{m_1+m_2, \beta_{0}^{2}+m_2\}$ is given by the sections of these curves that lie within the alternate conic section, i.e. the arcs of the circle that lie within the ellipse, and the arcs of the ellipse that lie within the circle. To see this, consider a point $p$ on the curve $\mathcal{M}=k$. $p$ must then be on either the circle or the ellipse, whichever has maximal value at that point. This requires that the \textit{other} curve has value less than $k$ at the point $p$. Curves of constant $B$ are given by circles centred at the origin, with increasing radii for increasing constant value. The same is true for curves of constant $\beta_{0}^{2}+m_2$, i.e. ellipses centred at the origin with increasing semi-major and semi-minor axes for increasing constant value. Thus, for one curve to have value $k$ and one curve to have value less than $k$ at the point $p$, we require that $p$ is contained on one curve and \textit{inside} the other, for any $p$ on the curve $\mathcal{M}=k$. Thus, it is given by arcs of the circle contained within the ellipse, and vice-versa.\\
\begin{figure}
    \centering
    \includegraphics[width=12cm]{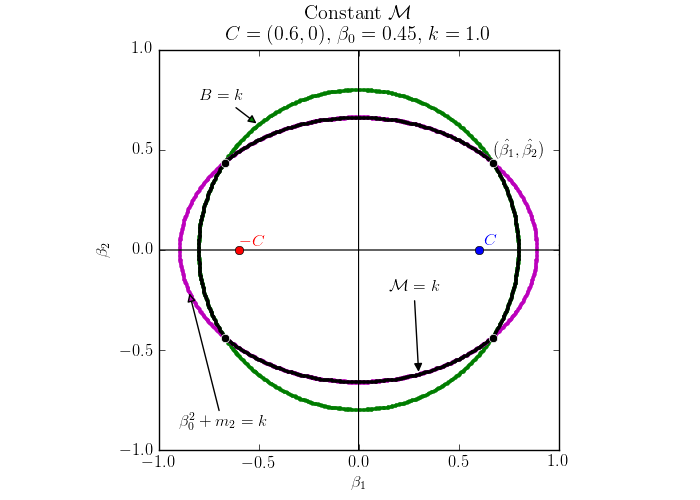}
    \caption{Curves of constant $\mathcal{M}=k$ for $k=1$, $C = (\beta_{4},\beta_{3}) = (0.6,0)$, $\beta_{0} = 0.45$. The curve of constant $B = 1$ is given by the green circle with radius $r_{B} = \sqrt{1 - \beta_{4}^{2}}=0.8$. The curve of constant $\beta_{0}^{2} + m_2 = 1$ is given by the pink ellipse with focii at $\pm C$ and semi-major and semi-minor axes given by (respectively) $a = \sqrt{m_2}\approx 0.893$, $b = \sqrt{m_2 - \beta_{4}^{2}}\approx 0.661$. The curve of constant $\mathcal{M}=1$ is given by the union of the arcs of the circle contained within the ellipse, and the arcs of the ellipse contained within the circle (black curve). The intersection points are given by $(\pm \hat{\beta_{1}}, \pm \hat{\beta_{2}})$.}
    \label{fig:constant_nonlocality}
\end{figure}
\end{proof}
The set of locally realistic states then lie in the closed region bounded by $\mathcal{V}$ and the curve given by $\mathcal{M}=1$ (see Fig. \ref{fig:nonlocality}, and note that $\beta_{3} \neq 0$ in this graph).\\
\\
The four intersection points (when they exist) of these curves are found to be at $(\pm \hat{\beta_{1}}, \pm \hat{\beta_{2}})$, where 
 \begin{subequations}
 \begin{equation*}
     \hat{\beta_{1}} = \sqrt{m_1 m_2}/\beta_{4},
 \end{equation*}
 \begin{equation*}
     \hat{\beta_{2}} = \sqrt{-(m_1 - \beta_{4}^{2})(m_2 - \beta_{4}^{2})}/\beta_{4}.
 \end{equation*}
 \end{subequations}
As the circle and ellipse are both centred on the origin, they intersect only when the circle radius $r_{B}=\sqrt{k-\beta_{4}^{2}}$ is valued between the semi-major and semi-minor axes $a,b$. This can be expressed in terms of $k$:
\begin{equation}
    m_2 \leq k \leq m_2 + \beta_{4}^{2}.
\end{equation}
One can immediately note that $k=m_1 + m_2$ cannot fall below the lower bound above. For $k$ above the upper bound, the circle $B=k$ entirely surrounds the ellipse, and the curve $\mathcal{M}=k$ is given by only the ellipse.\\
\\
For Group 2 $X$-states with general $\tau$'s, one can compute an upper bound for $\mathcal{M}$ using the validity criteria, which when written in terms of $B$ and $\text{det}M$ give us (for type \textsc{I} states):
\begin{subequations}\label{eq:validity_crit_B}
\begin{equation}
    \sqrt{B - 2\text{det}M + (\tau_{1}+\tau_{2})^{2}} \leq 1 + \beta_{0},
\end{equation}
\begin{equation}
    \sqrt{B + 2\text{det}M + (\tau_{1}-\tau_{2})^{2}} \leq 1 - \beta_{0}.
\end{equation}
\end{subequations}
Combining these to get conditions on $B$ and $U$ and taking sign transformations on $\tau_{2}$ and $\beta_{0}$ to account for type \textsc{II} states, one obtains
\begin{equation*}
    B \leq 1 + \beta_{0}^{2} - (\tau_{1}^{2} + \tau_{2}^{2}),
\end{equation*}
\begin{equation*}
    U^{2} \leq (1-\beta_{0}^{2})^{2} -2B(\tau_{1}^{2} + \tau_{2}^{2}) - (-1)^{t}8\tau_{1}\tau_{2}\text{det}M,
\end{equation*}
which give the following upper bound on $\mathcal{M}$:
\begin{equation}
    \mathcal{M} \leq \beta_{0}^{2} + \frac{1}{2}\left(1 + \beta_{0}^{2} - (\tau_{1}^{2}+\tau_{2}^{2}) + \sqrt{(1-\beta_{0}^{2})^{2} - 2B(\tau_{1}^{2}+\tau_{2}^{2}) -(-1)^{t} 8\tau_{1}\tau_{2}\text{det}M } \right).
\end{equation}
When $\tau_{i}=0$ this reduces to $\mathcal{M}_{\tau=0} \leq 1 + \beta_{0}^{2}$, which also holds for the nine Mermin-states $\mathcal{Q}_{i}$, as they differ from Group 2 $X$-states only by their $\tau$-parameters. Saturating this bound to the maximal case $\mathcal{M}=2$ is a necessary and sufficient condition for purity, as can be shown:
\begin{prop}
For $\rho$ being a Group 2 $X$-state or one of the nine $\mathcal{Q}_{i}$-states with $\tau_i = 0$, it violates the Bell Inequality maximally iff is pure.
\end{prop}
\begin{proof}
It is known \cite{horodecki2001separability} that a state written in Pauli operator form is pure (i.e. $\text{Tr}(\rho^{2})=1)$ iff $\sum_{i}\tau_{i}^{2} + \sum_{i,j}\beta_{ij}^{2} = 3$. For our considerations this reduces to $\beta_{0}^{2} + B = 3$ (recall that $B = \sum_{i=1}^{4}\beta_{i}^{2}$).\\
($\Rightarrow$) Allowing the state to be maximally nonlocal, one has that $\mathcal{M}=2$ and thus by the upper bound given above, $\beta_{0}^{2} = 1$. For the former case in \eqref{eq:nonlocality}, maximal nonlocality gives $B=2$ and thus purity. For the latter case, the validity conditions \eqref{eq:validity_crit_B} with $|\beta_{0}|=1$ give $B = \pm 2\text{det}M$, and then $U=0$. Then one has
\begin{equation*}
    \mathcal{M} = \beta_{0}^{2} + \frac{1}{2}(B+U),
\end{equation*}
\begin{equation*}
    \Rightarrow 2 = 1 + B/2 \Rightarrow B=2,
\end{equation*}
giving $\beta_{0}^{2} + B = 3$.\\
($\Leftarrow$) Taking the former case in \eqref{eq:nonlocality}, one can write $\mathcal{M} = B = 3 - \beta_{0}^{2}$. It is known that $\mathcal{M}$ is at most 2 \cite{horodecki1995violating} and it can be seen by the validity conditions that $|\beta_{0}|$ is at most 1. $\mathcal{M}=3-\beta_{0}^{2}$ then requires that $\beta_{0}^{2} = 1$, and $\mathcal{M}=2$. For the latter case in \eqref{eq:nonlocality}, one squares and adds the validity conditions \eqref{eq:validity_crit_B} to get
\begin{equation*}
    1 + \beta_{0}^{2} \geq B = 3 - \beta_{0}^{2},
\end{equation*}
\begin{equation*}
    \Rightarrow \beta_{0}^{2} \geq 1,
\end{equation*}
as $|\beta_{0}|$ is at most 1, one recovers that $\beta_{0}^{2}=1$ and thus $B=2$. By the previous argument, $\beta_{0}^{2}=1 \Rightarrow U=0$ and then $\mathcal{M} = 1 + 2/2 = 2$.
\end{proof}

One thus recovers the fact that the known examples of two-qubit states that maximally violate Bell-inequality are pure \cite{horodecki1995violating} and, as also indicated by several other studies, that for mixed states a large amount of entanglement seems to be necessary to obtain violation of Bell-inequalities \cite{munro2001}. Using our description of $X$-states of Group 2 and Mermin states with generalised parameters, one can represent the region of entanglement and violation of Bell-inequality in the parameter space of  Figures \ref{fig:Validity_separability_graph} and \ref{fig:validity_separability_graph_colours_main} as depicted in Figure 11:
\begin{figure}[!h] \centering
 \includegraphics[width=13cm]{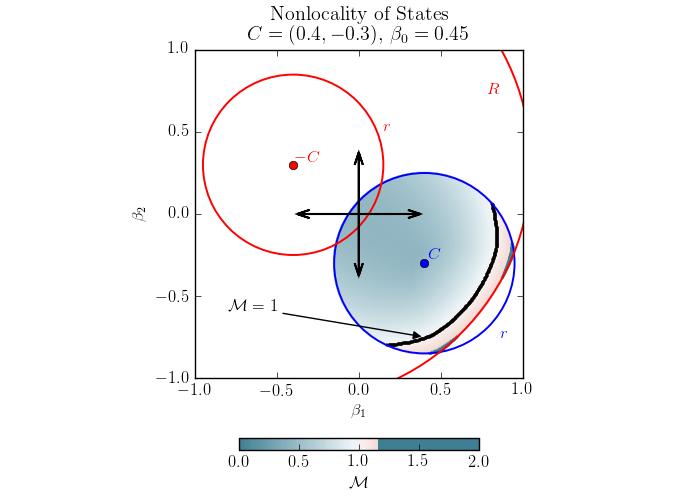}
 \caption{Graph of nonlocality measure $\mathcal{M}$ for various hyperplane states with $\tau_i=0$, $C := (\beta_{4}, \beta_{3}) = (0.4,-0.3)$ and $\beta_0=0.45$. Values of $\beta_1$ and $\beta_2$ are plotted on the $x,y$-axes, respectively. Values for $\mathcal{M}$ are computed for states within the validity region $\mathcal{V}=(C,r)\cap(-C,R)$ and given by the red-blue scale. A set of states parametrised by constant $\mathcal{M}$ forms an ellipse with focal points at $C$, $-C$ in the plane, the example $\mathcal{M}=1$ shown in black (for the subset within $\mathcal{V}$). As can be seen, only a small proportion of valid states are nonlocal $(\mathcal{M}>1)$, and a high amount of entanglement is needed to satisfy this condition.}
 \label{fig:nonlocality}
\end{figure}

\section{Conclusion}
In this paper we have considered the $15$ types of $X$-states following a remark of Rau \cite{rau2009algebraic} who suggested to define $X$-states from the algebraic structure of the defining two-qubit operators used to decompose the state in the Pauli basis. We showed that these $15$ types split in two groups $9+6$ when we consider their entanglement properties. The Group 2 ($9$ types) is the only one that can produce entangled states. A uniform treatment of those states is proposed and some criteria in terms of their parameters have been proposed for states with maximally-mixed subsystems. The introduction of those parameters allows us to give a representation of the validity, entanglement and Bell-violation area. One also proposed an alternative geometric definition of $X$-states as perp-hyperplanes of the symplectic polar space $\mathcal{W}(3,2)$. This new definition establishes an interesting connection between the finite geometric picture introduced in quantum information to study configuration of two-qubit operator and the two-qubit density matrices. In particular, we started to study hyperplanes-states, i.e. density matrices which involve two-qubit operators defining a geometric hyperplane. In this line many questions related to the concepts of Veldkamp line and Veldkamp space, the space of geometric hyperplanes \cite{Saniga2007}, can be addressed. When we restrict ourself to density matrices with maximally-mixed subsystems one observed that all hyperplane-states of $\mathcal{W}(3,2)$ coincide with the generalised $X$-states with maximally mixed subsystems. However, if one does not consider the restriction on the subsystems, Mermin-states and $X$-states are not equivalent anymore and it could be interesting in a future work to differentiate those states in terms of their quantum properties. Another direction would be to consider three-qubit $X$-states where similar algebraic structure show up \cite{vinjanampathy2010generalized} and where the finite geometric picture of the three-qubit Pauli group reveals fascinating properties \cite{levay2017magic,saniga2020magic}.

\section*{Acknowledgments}
This work was supported, in part, by the Slovak Research and Development Agency under the contract $\#$ SK-FR-2017-0002, as well as by the Slovak VEGA Grant Agency, Project $\#$ 2/0004/20. This work was also supported by the National Research Development and Innovation Office of Hungary within the Quantum Technology National Excellence Program (Project No. 2017-1.2.1-NKP-2017-0001). In France the project was supported by the French ``Investissements d'Avenir'' programme, project ISITE-BFC (contract ANR-15-IDEX-03) and the EUR-EIPHI Graduate School (Grant No. 17-EURE-0002).

\appendix
\section{The $15$ Fano planes of PG$(3,2)$}\label{appA}
The geometry of PG$(3,2)$ comprises $15$ points, $35$ lines and $15$ planes. All the $15$ planes corresponding to $15$ types of $X$-states can be labelled by the $15$ points as explained in the introduction:
\begin{equation}
 F_p=\{q\in \text{PG}(3,2),\sigma(p,q)=0\}.
\end{equation}
For the convenience of the reader we provide in Table \ref{tab:fano} the explicit list of the $15$ Fano planes of PG$(3,2)$ with respect to the Group $1/2$ splitting.
\begin{table}[!h]
\centering
 \begin{tabular}{|c|c|c|}
 \hline
  Group & $p\in \text{PG}(3,2)$ & $F_p$ \\
  \hline 
  $1$   & $[0:0:0:1] \leftrightarrow IX$ & $XX,YX,ZX,XI,YI,ZI,IX$\\
        & $[0:0:1:0] \leftrightarrow IZ$ & $XZ,YZ,ZZ,XI,YI,ZI,IZ$\\
        & $[0:0:1:1] \leftrightarrow IY$ & $XY,YY,ZY,XI,YI,ZI,IY$\\
        & $[0:1:0:0]\leftrightarrow XI$ & $XX,XY,XZ,IX,IY,IZ,XI$\\
        & $[1:0:0:0]\leftrightarrow ZI$ & $ZX,ZY,ZZ,IX,IY,IZ,ZI$\\
        & $[1:1:0:0]\leftrightarrow YI$ & $YX,YY,YZ,IX,IY,IZ,YI$\\
        \hline
   $2$  &  $[0:1:0:1] \leftrightarrow XX$ & $XX,YY,YZ,ZY,ZZ,IX,XI$\\
         &  $[0:1:1:0] \leftrightarrow XZ$ & $XZ,YY,YX,ZX,ZY,IZ,XI$\\
          &  $[0:1:1:1] \leftrightarrow XY$ & $XY,YX,YZ,ZX,ZZ,IY,XI$\\
           &  $[1:0:0:1] \leftrightarrow ZX$ & $ZX,XY,XZ,YY,YZ,IX,ZI$\\
            &  $[1:0:1:0] \leftrightarrow ZZ$ & $ZZ,YY,XX,XY,YX,IZ,ZI$\\
             &  $[1:0:1:1] \leftrightarrow ZY$ & $ZY,XX,XZ,YX,YZ,IY,ZI$\\
              &  $[1:1:0:1] \leftrightarrow YX$ & $YX,XY,XZ,ZY,ZZ,IX,YI$\\
               &  $[1:1:1:0] \leftrightarrow YZ$ & $YZ,XX,XY,ZX,ZY,IZ,YI$\\
                &  $[1:1:1:1] \leftrightarrow YY$ & $YY,XX,XZ,ZX,ZZ,IY,YI$\\
                \hline
 \end{tabular}
\caption{The $15$ Fano planes of PG$(3,2)$ according to the Group 1/2 splitting.}\label{tab:fano}
\end{table}

\bibliographystyle{plain}

\bibliography{biblio}

\end{document}